\documentclass{article}
\usepackage{graphicx}
\graphicspath{{./pictures/}}
\usepackage{amsmath,amsfonts,amsthm}

\usepackage{url}

\allowdisplaybreaks

\newcommand{\R}{\mathbb{R}}

\newcommand{\N}{\mathbb{N}}

\newcommand{\E}{\mathcal{E}}

\DeclareMathOperator{\sech}{sech}

\newtheorem{Definition}{Definition}

\newtheorem{Theorem}{Theorem}

\title{Applications of nonlocal constants of motion\\
in Lagrangian Dynamics}
\author{Gianluca Gorni\\
Universit\`a di Udine\\
Dipartimento di Matematica e Informatica\\
via delle Scienze~208, 33100 Udine, Italy\\
\tt{gianluca.gorni@uniud.it}
\and
Gaetano  Zampieri\\
Universit\`a di Verona\\
Dipartimento di Informatica\\
strada Le Grazie 15, 37134 Verona, Italy\\
\tt{gaetano.zampieri@univr.it}}

\date{}

\begin{document}

\maketitle

\begin{abstract}
We give a recipe to generate ``nonlocal'' constants of motion for ODE Lagrangian systems and we apply the method to find useful constants of motion for dissipative system, for the Lane-Emden equation, and for the Maxwell-Bloch system with RWA.
\end{abstract}

\section{Introduction}
\label{Noether}

Analytical Dynamics in one independent variable~$t$ studies those systems whose ``natural'' motions $t\mapsto q(t)\in\R^n$ are fixed-extrema stationary points for some action functional
\begin{equation}\label{actionFunctional}
  S_{a,b}\bigl(q(\cdot)\bigr):=
  \int_a^b
  L\bigl(t,q(t),\dot q(t)\bigr)\,dt
\end{equation}
in the sense of the Calculus of Variations, where the scalar function $L(t,q,\dot q)$ is called the Lagrangian of the system. Under a modicum of regularity assumptions, the natural motions are the solutions to the \emph{Euler-Lagrange equation}
 \begin{equation}\label{Euler-Lagrange}
  \frac{d}{dt}\partial_{\dot q}L\bigl(t,q(t),\dot q(t)\bigr)-
  \partial_qL\bigl(t,q(t),\dot q(t)\bigr)
  =0 \qquad\text{for all }t.
\end{equation}

Here and in the sequel we will use the notation $\partial_{q}$ (or occasionally $\nabla$) for the partial derivative (gradient) with respect to the vector~$q$, the symbol $\partial_{\dot q}$ for the partial derivative with respect to the vector~$\dot q$, and the notation $\cdot$~for the scalar product in~$\R^n$.

A~first integral for the mechanical system is a function of the form
\begin{equation}\label{trueFirstIntegral}
  N(t,q,\dot q),\qquad t\in\R,\quad q,\dot q\in\R^n,
\end{equation}
that is constant along all natural motions. Noether's Theorem~\cite{NoetherOriginal} establishes a connection between first integrals and symmetry properties of the Lagrangian~$L$.

A previous work~\cite{GZNoether} of ours revisited Noether's Theorem from different points of view, such as how to interpret asynchronous perturbations (or ``time change'') and boundary (or Bessel-Hagen, or gauge) terms (see~\cite[Sec.~4--5]{GZsao} for a partial, simplified version). Here we take up the part of that paper where we extended the theory so as to include not just ``true'' first integrals in the sense above~\eqref{trueFirstIntegral}, but also constants of motion of the form
\begin{equation}\label{genericNonlocaConstantOfMotion}
  N\bigl(t,q(t),\dot q(t)\bigr)+
  \int_{t_0}^t M\bigl(s,q(s),\dot q(s)\bigr)ds\,.
\end{equation}
Such a constant of motion will be called ``nonlocal'', or integral, in the sequel, because its value at a time~$t$ depends not only on the value of position and velocity at time~$t$, but also on the past history of the motion.

The basic, very simple result on nonlocal constants of motion in that paper~\cite[Subs.~3.2]{GZNoether} can be reformulated in the following self-contained way, which is all that is needed for the sequel:

\begin{Theorem}\label{theoremOfVeryGeneralConstantOfMotion}
Let $t\mapsto q(t)$ be a solution to the Euler-Lagrange equation~\eqref{Euler-Lagrange} and $q_\lambda(t)$ be a smooth family of synchronous perturbed motions, with the parameter $\lambda$ in a neighbourhood of $0\in\R$, and such that $q_\lambda(t)\equiv q(t)$ when $\lambda=0$. Then the following function is constant:
\begin{equation}\label{veryGeneralConstantAlongMotion}
  t\mapsto
  \partial_{\dot q}
  L\bigl(t,q(t),\dot q(t)\bigr)\cdot
  \partial_\lambda q_\lambda(t)
  \big|_{\lambda=0}-
  \int_{t_0}^t
  \frac{\partial}{\partial\lambda}
  L\bigl(s,q_\lambda(s),\dot q_\lambda(s)\bigr)
  \Big|_{\lambda=0}ds\,.
\end{equation}
\end{Theorem}

\begin{proof}
Simply take the time derivative of~\eqref{veryGeneralConstantAlongMotion}, reverse the derivation order and use the Euler-Lagrange equation~\eqref{Euler-Lagrange}.
\end{proof}

The Theorem above justifies the following terminology, that we will use repeatedly in this work:

\begin{Definition}
\label{definitionOfSyncronousFamily}
For a given natural motion $q({}\cdot{})$, a one parameter family of perturbed motions (or, simply, ``a~family'') will be a smooth function $(\lambda,t)\mapsto q_{\lambda}(t)$ such that $q_0(t)\equiv q(t)$, and the nonlocal constant of motion associated with $q_\lambda$ will be formula~\eqref{veryGeneralConstantAlongMotion}.
\end{Definition}

Theorem~\ref{theoremOfVeryGeneralConstantOfMotion} gives us a simple machinery that takes a perturbed family such as
\begin{equation}
  q(t+\lambda),\qquad
  q(t+\lambda e^{at}),\qquad
  q(t+\lambda t^2),\qquad
  e^\lambda q(e^\lambda t) \text{ etc.}
\end{equation}
and computes its associated nonlocal constant of motion. If the family is chosen randomly the constant of motion can very well turn out to be trivial (a numeric constant, for example) or of no apparent practical value. The original work~\cite{GZNoether} provided a small early sample of systems where such constant of motion seemed interesting. Our purpose in this paper is to pursue the topic much further, showing how carefully selected families can lead to nonlocal constants of motion that are useful in studying the system.

Our initial guiding idea here is that a function such as~\eqref{genericNonlocaConstantOfMotion} has a chance to be valuable if the the integrand~$M$ does not change sign, because then $N$~will be monotonic in time; if, in addition, $N$ happened to be coercive in~$q,\dot q$, or at least in~$\dot q$, we could derive estimates that imply global existence in time for either the past or the future. To massage formula~\eqref{veryGeneralConstantAlongMotion} into such a useful form we will replace~$\ddot q$ with the Euler-Lagrange equation~\eqref{Euler-Lagrange} and integrate by part as needed.

In Section~\ref{dissipativesection} we use nonlocal constants of motion to prove global existence and asymptotic estimates for the solutions of the dissipative equation $\ddot q=-k\dot q-\nabla U(q)$ when $U$ is a potential bounded from below. In particular, the existence in the past has not been considered before, to the best of our knowledge.

In Section~\ref{LaneEmdenSection} we provide three families that together prove the global existence of solution for the Lane-Emden equation when $n$ is odd, with asymptotic estimates as~$t\to+\infty$ when $n$ is odd and~$\ge5$.

Section~\ref{MBsection} is about the Maxwell-Bloch with rotating wave approximation (RWA) system in Ca\c{s}u's Lagrangian formulation~\cite{Casu}. We show a family that is not directly useful in the sense that we described, but that leads to a separation of one of the variables from the others, and to a time-dependent true first integral.

\section{Dissipative mechanical system}\label{dissipativesection}

Next, let us consider the equation of motion
\begin{equation}\label{eqdissipative}
  \ddot q=-k\dot q-\nabla U(q),\qquad q\in\R^n,
\end{equation}
where $k>0$ and $U$ is a smooth potential on~$\R^n$. To study the system it is natural to use to the function 
\begin{equation}
  \mathcal{E}(q,\dot q)=\frac{1}{2} \lvert\dot q\rvert^2+U(q)
\end{equation}
which is the energy first integral in the non-dissipative case $k=0$. It is easy and well-known to prove that $\mathcal{E}$ is \emph{decreasing} along solutions in the actual dissipative case $k>0$:
\begin{equation}\label{derivativeOfE}
  \dot{\mathcal{E}}=
  \dot q\cdot \partial_{q}\mathcal{E}
  +\ddot q\cdot\partial_{\dot q}\mathcal{E}=
  \dot q\cdot \nabla U(q)
  +\Bigl(-k\dot q-\nabla U(q)\Bigr)\cdot\dot q
  =-k\lvert\dot q\rvert^2\le 0
\end{equation}
where $\partial_{q}\E$ is the gradient in $q$, $\partial_{\dot q}\E$ the one in $\dot q$, and $\cdot$ denotes the scalar product. If we assume that $U$ is bounded from below, i.e., that
\begin{equation}\label{coerciveDefinition}
  U_{\inf}:=\inf_{q\in\R^n} U(q)>-\infty,
\end{equation}
then the fact that $\mathcal{E}$ decreases along the solution implies that $\dot q(t)$ is bounded in the future:
\begin{equation}\label{boundedInTheFuture}
  \frac{1}{2} \bigl|\dot q(t)\bigr|^2\le
  \frac{1}{2} \bigl|\dot q(t_0)\bigr|^2+U\bigl(q(t_0)\bigr)
  -U_{\inf}
  \quad\text{for }t\ge t_0.
\end{equation}
With a standard reasoning, $q(t)$ is uniformly continuous in the future, hence bounded when $t$ is bounded, and we can conclude that the solutions are global in the future.

A complementary information can be obtained if we reinterpret equation~\eqref{derivativeOfE} as asserting that the following expression
\begin{equation}\label{nonlocalEnergyDerivative}
  \frac{1}{2} \bigl|\dot q(t)\bigr|^2+U\bigl(q(t)\bigr)
  +k\int_{t_0}^t\bigl|\dot q(s)\bigr|^2ds.
\end{equation}
is a nonlocal constant of motion of the form~\eqref{genericNonlocaConstantOfMotion}. This implies that
\begin{multline}
  U_{\inf}+k\int_{t_0}^t\bigl|\dot q(s)\bigr|^2ds
  \le 
  \frac{1}{2} \bigl|\dot q(t)\bigr|^2+U\bigl(q(t)\bigr)
  +k\int_{t_0}^t\bigl|\dot q(s)\bigr|^2ds=\\
  =\frac{1}{2}\bigl|\dot q(t_0)\bigr|^2
  +U\bigl(q(t_0)\bigr) \quad\text{for all }t,
\end{multline}
from which we can extract an $L^2$ estimate of~$\dot q$ in the future:
\begin{equation}\label{L2estimateInTheFuture}
  \int_{t_0}^{+\infty}\bigl|\dot q(s)\bigr|^2ds<+\infty.
\end{equation}

Our second order equation~\eqref{eqdissipative} is the Euler-Lagrange equation of the Lagrangian
\begin{equation}
   L(t,q,\dot q):=e^{kt}\Bigl(
   \frac{1}{2}\lvert\dot q\rvert^2-U(q)\Bigr).
\end{equation}
In our paper \cite{GZNoether} we found a non-local constant of motion by means of the  the time-shift family of perturbed motions $q_\lambda(t):=q(t+\lambda)$ of the natural $q(t)$. This choice is quite natural since it yields the conservation of energy in the conservative case $k=0$. Now, let us consider the more general  family
\begin{equation}\label{familydissipative}
  q_\lambda(t):=q(t+\lambda e^{at})
\end{equation}
with  a new real parameter~$a$. Then, using the Lagrange equation~\eqref{eqdissipative}
\begin{align*}
  \frac{\partial}{\partial\lambda}
  L\bigl(t,&q_\lambda(t),\dot q_\lambda(t)\bigr)
  \Big|_{\lambda=0}=\\
  &=e^{kt}\frac{\partial}{\partial\lambda}\Bigl(
  \frac{1}{2}\bigl((1+\lambda ae^{at})
  \dot q(t+\lambda e^{at})\bigr)^2
  -U\bigl(q(t+\lambda e^{at})\bigr)\Bigr) \Big|_{\lambda=0}=\\
  &=e^{(a+k)t}\dot q(t)\cdot\Bigl(a\dot q(t)
  -\nabla U\bigl(q(t)\bigr)+\ddot q(t)\Bigr)
  =\\
  &=e^{(a+k)t}\dot q(t)\cdot\Bigl((a-k)\dot q(t)-
  2\nabla U\bigl(q(t)\bigr)\Bigr)=\\
  &=\frac{d}{dt}\left(-2e^{(a+k)t}U\bigl(q(t)\bigr)\right)+\\
  &\qquad+e^{(a+k)t}\Bigl((a-k)\lvert\dot q(t)\rvert^2
  +2(a+k)U\bigl(q(t)\bigr)\Bigr).
\end{align*}
The associated nonlocal constant of motion is
\begin{multline}\label{energyPluspositiveintegral}
  e^{(a+k)t}\Bigl(\lvert\dot q(t)\rvert^2
  +2U\bigl(q(t)\bigr)\Bigr)
  +{}\\
  +\int_t^{t_0}e^{(a+k)s}\Bigl((a-k)\lvert\dot q(s)\rvert^2
  +2(a+k)U\bigl(q(s)\bigr)\Bigr)ds,
\end{multline}
or, equivalently,
\begin{multline}\label{energyPluspositiveintegral2}
  e^{(a+k)t}\Bigl(\lvert\dot q(t)\rvert^2
  +2U\bigl(q(t)\bigr)-2U_{\inf}\Bigr)
  +2e^{(a+k)t_0}U_{\inf}
  -{}\\
  -\int_{t_0}^t e^{(a+k)s}
  \Bigl((a-k)\lvert\dot q(s)\rvert^2
  +2(a+k)\bigl(U\bigl(q(s)\bigr)-U_{\inf}\bigr)\Bigr)ds,
\end{multline}
If we choose either $a\le-k$ or $a\ge k$ the integrand in formula~\eqref{energyPluspositiveintegral2} does not change sign. Actually, the choice $a=-k$ makes the constant of motion~\eqref{energyPluspositiveintegral2} simply a multiple of the one given by formula~\eqref{nonlocalEnergyDerivative}, from which we have extracted information on the solutions in the future. With the alternative choice $a\ge k$, the integrand in~\eqref{energyPluspositiveintegral} is~$\ge0$, so that the integral term is monotonically increasing in~$t$ and the function
\begin{equation}
  e^{(a+k)t}\Bigl(\bigl|\dot q(t)\bigr|^2
  +2U\bigl(q(t)\bigr)-2U_{\inf}\Bigr)=
  2e^{(a+k)t}\bigl(\mathcal{E}\bigl(q(t),\dot q(t)\bigr)
  -U_{\inf}\bigr)
\end{equation}
must be monotonically increasing. With a standard reasoning we can conclude that all solutions are global in the past too, and we have the estimates, choosing $a=k$,
\begin{gather}
  \mathcal{E}\bigl(q(t),\dot q(t)\bigr)-U_{\inf}\le
  e^{2k(t_0-t)}\Bigl(
  \mathcal{E}\bigl(q(t_0),\dot q(t_0)\bigr)
  -U_{\inf}\Bigr)\quad\text{for all }t\le t_0,
  \label{energyEstimateInThePast}\\
  \bigl|\dot q(t)\bigr|^2\le e^{2k(t_0-t)}
  \Bigl(\bigl|\dot q(t_0)\bigr|^2+2U\bigl(q(t_0)\bigr)
  -2U_{\inf}\Bigr)\quad\text{for all }t\le t_0,
  \label{speedEstimateInThePast}
\end{gather}
which are sharp in the trivial case when $U$ is constant.

Summing up:

\begin{Theorem}
If $k>0$ and $U$ is a smooth potential on~$\R^n$ which is bounded from below, all solutions of the dissipative equation $\ddot q=-k\dot q-\nabla U(q)$ are defined for all~$t\in\R$, and we have the estimates~\eqref{boundedInTheFuture}, \eqref{L2estimateInTheFuture} in the future and~\eqref{energyEstimateInThePast}, \eqref{speedEstimateInThePast} in the past.
\end{Theorem}

Notice that for $a=0$ the nonlocal constant of motion~\eqref{energyPluspositiveintegral} becomes
\begin{equation}\label{energyPlusAction}
  2 E\bigl(t,q(t),\dot q(t)\bigr)+
  2 k\int_{t_0}^t
  L\bigl(s,q(s),\dot q(s)\bigr)ds\,,
\end{equation}
where we recognize the action integral multiplied by $2k$, and $E$ is the dissipative energy  
\begin{equation}\label{dissipativeEnergy}
  E(t,q,\dot q)=
  \partial_{\dot q}L(t,q,\dot q)\,\dot q(t)-
  L(t,q,\dot q)=e^{kt}{\mathcal E}(q,\dot q)\,.
\end{equation}

\section{The Lane-Emden equation}
\label{LaneEmdenSection}

The Lane-Emden system has the following Lagrangian function and Euler-Lagrange equation:
\begin{gather}\label{lagrangianLaneEmden}
  L(t,q,\dot q):=t^2\Bigl(\frac{\dot q^2}{2}-\frac{q^{n+1}}{n+1}
  \Bigr),\\
  \ddot q=-q^n-\frac{2}{t}\dot q,\label{eqlane}
\end{gather}
where $q\in\R$, $n\in\N$, $t>0$. This  is a  form of Poisson's equation for the hydrostatic equilibrium of a self-gravitating spherically symmetric fluid, arising in the study of stellar interiors~\cite{Chandrasekhar}. It is named after astrophysicists Jonathan Homer Lane and Robert Emden. For such purpose the parameter~$t$ does not mean time but rather the distance from the center and $q$~is related to density and pressure. The most common boundary conditions are $q(0)>0$, $\dot q(0)=0$, and for these the literature has sharp results of global existence of solutions under such boundary conditions. Lane-Emden equation and its generalisations has been also applied in other branches of physics, for instance in  kinetic theory and quantum mechanics. For more information see Goenner and Havas' article~\cite{GoennerHavas} and the references therein. Our contribution in this Section is to exhibit three nonlocal constants of motion for the Lane-Emden equation which, when $n$ is odd, have a bearing on the maximal interval of existence for Cauchy initial conditions at some $t_0>0$, and on the asymptotic behaviour of solutions as $t\to+\infty$.

The first family is
\begin{equation}\label{lamdEndenFirstFamily}
  q_\lambda(t):=
  q\Bigl(t-\frac{\lambda}{t^2}\Bigr).
\end{equation}
We can compute
\begin{align*}
  \frac{\partial}{\partial\lambda}
  L\bigl(t,&q_\lambda(t),\dot q_\lambda(t)\bigr)
  \Big|_{\lambda=0}=\\
  ={}&t^{2}\frac{\partial}{\partial\lambda}
  \Bigl(\frac{1}{2}\bigl((1+2\lambda t^{-3})
  \dot q(t+\lambda t^{-2})\bigr)^2
  -\frac{1}{n+1}q(t+\lambda t^{-2})^2\Bigr)
  \Big|_{\lambda=0}=\\
  ={}&-q(t)^n\dot q(t)+\dot q(t)\ddot q(t)-\frac{2\dot q(t)}{t}=\\
  ={}&
  \frac{d}{dt}\Bigl(\frac{1}{2} \dot q(t)^2
  -\frac{1}{n+1} q(t)^{n+1}\Bigr)
  +\frac{2}{t}\dot q(t)^2.
\end{align*}
The associated nonlocal constant of motion is
\begin{equation}\label{nonlocalConstant1}
  \frac{1}{2} \dot q(t)^2+\frac{1}{n+1} q(t)^{n+1}
  +\int_{t_0}^t\frac{2}{s}\dot q(s)^2ds.
\end{equation}
You will recognize that the first two terms in the sum
\begin{equation}\label{calE}
  \mathcal{E}(q,\dot q):=\frac{1}{2}\dot q^2+\frac{1}{n+1} q^{n+1}.
\end{equation}
are what would be the energy first integral if the dissipative term $-2\dot q/t$ were deleted from the Lagrange equation. This~$\mathcal{E}$ decreases along solutions, because the integrand in~\eqref{nonlocalConstant1} is nonnegative. Moreover, when $n$ is odd, this $\mathcal{E}$ is a coercive norm in~$\R^2$. From this we can deduce global existence and boundedness in the future for all solutions to Lane-Emden equations. Of~course, $\mathcal{E}$~is an obvious quantity to watch for, and the fact that it is decreasing can be established easily without appealing to family~\eqref{lamdEndenFirstFamily} and integral constant~\eqref{nonlocalConstant1}.
 
To deal with existence the past, let us consider the following family
\begin{equation}
  q_\lambda(t):=q(t+\lambda t^2).
\end{equation}
We can compute, using the Lagrange equation~\eqref{eqlane}
\begin{align*}
  \frac{\partial}{\partial\lambda}
  L\bigl(t,&q_\lambda(t),\dot q_\lambda(t)\bigr)
  \Big|_{\lambda=0}=\\
  &=t^{2}\frac{\partial}{\partial\lambda}
  \Bigl(\frac12\bigl((1+2\lambda t)
  \dot q(t+\lambda t^2)\bigr)^2
  -\frac{1}{n+1}q(t+\lambda t^2)^2\Bigr)
  \Big|_{\lambda=0}=\\
  &=t^3\dot q(t)\Bigl(2\dot q(t)-t q(t)^n+t\ddot q(t)\Bigr)
  =\\
  &=-2t^4q(t)^n\dot q(t)=
  -\frac{d}{dt}\Bigl(\frac{2}{n+1}t^4q(t)^{n+1}\Bigr)
  +\frac{8}{n+1}t^3q(t)^{n+1}.
\end{align*}
The associated constant of motion is 
\begin{equation}\label{energylanePluspositiveintegral}
  t^4\Bigl(\dot q(t)^2+\frac{2}{n+1}q(t)^{n+1}\Bigr)
  +\frac{8}{n+1}\int_t^{t_0}s^3q(s)^{n+1}.
\end{equation}
When $n$ is odd the integrand is~$\ge0$, and the function
\begin{equation}
  t^4\Bigl(\dot q^2+\frac{2}{n+1}q^{n+1}\Bigr)
\end{equation}
is coercive when $t$ is away from~0. This guarantees that 0 is the infimum of the maximal time interval of existence for a solution.

The third family of perturbed motions was already introduced in our previous paper~\cite[Sec.~8]{GZNoether}:
\begin{equation}\label{thirdFamilyForLaneEmden}
  q_\lambda(t):= e^\lambda q(e^\lambda t)
\end{equation}
We can compute, again using the Lagrange equations~~\eqref{eqlane}
\begin{equation*}
  \frac{\partial}{\partial\lambda}
  L\bigl(t,q_\lambda(t),\dot q_\lambda(t)\bigr)
  \Big|_{\lambda=0}
  =\frac{d}{dt}\biggl(-\frac{2}{n+1}t^3q(t)^{n+1}\biggr)
  {}+\frac{5-n}{n+1}\,t^2q(t)^{n+1}.
\end{equation*}
whence the constant of motion
\begin{equation}\label{constantLane-Emden}
  t^2 \Bigl(\frac{2}{n+1} t q(t)^{n+1}+q(t)\dot q(t)+t
  \dot q(t)^2\Bigr)
  +\frac{n-5}{n+1} \int_{t_0}^t s^2 q(s)^{n+1} \,ds
\end{equation}
which is a first integral in the well-known case $n=5$ and a non-local constant of motion otherwise. When $n$ is odd and~$\ge5$ we can exploit the fact that the integrand is~$\ge0$ and that the function
\begin{equation}
  t^2 \Bigl(\frac{2}{n+1} t q^{n+1}+q\dot q+t
  \dot q^2\Bigr)
\end{equation}
is more and more coercive as~$t\to+\infty$. Let $t_0>0$ belong to the domain of the given solution $q(t)$. Since, as we noticed, the integrand $s^2 q(s)^{n+1}$ is nonnegative for all~$s$, there exists a constant~$c_1$ such that
\begin{equation}\label{c1}
  \frac{2}{n+1} t q(t)^{n+1}+q(t)\dot q(t)+t
   \dot q(t)^2\le \frac{c_1}{t^2}
\end{equation}
for all  $t\ge t_0$. On the other hand
\begin{align}
  \frac{2}{n+1} t q(t)^{n+1}+{}&q(t)\dot q(t)+t
   \dot q(t)^2\ge\notag\\
  &\ge \frac{2}{n+1} t q(t)^{n+1}-\frac{q(t)^2+\dot q(t)^2}{2}+t
   \dot q(t)^2=\label{c2}\\
   &=\frac{2}{n+1} t q(t)^{n+1}
   +\Bigl(t-\frac{1}{2}\Bigr)
   \dot q(t)^2-\frac{1}{2}q(t)^2.\notag
\end{align}
Combining the inequalities~\eqref{c1} and~\eqref{c2} we see that
\begin{equation*}
  \frac{2}{n+1} t q(t)^{n+1}
   +\Bigl(t-\frac{1}{2}\Bigr)
   \dot q(t)^2
   \le \frac{c_1}{t^2}+\frac{1}{2} q(t)^2\le c_2.
\end{equation*}
Hence for large~$t$
\begin{equation}\label{asymptoticEstimates}
  \lvert q(t)\rvert
   \le \frac{c_3}{t^{1/(n+1)}},\qquad
  \lvert\dot q(t)\rvert\le\frac{c_4}{\sqrt{t}}.
\end{equation}

Collecting the results:

\begin{Theorem}
If $n$ is odd then all solutions of Lane-Emden equation~\eqref{eqlane} are defined for all $t\in\mathopen\rbrack0,+\infty\mathclose\lbrack$. When $n$ is also~$\ge5$ the asymptotic estimates~\eqref{asymptoticEstimates} hold as $t\to+\infty$.
\end{Theorem}

Our first attempt with the Lane-Emden equation was with the family
\begin{equation}\label{familyWithDynamicalSymmetry}
  q_\lambda(t):= e^\lambda q(e^{\lambda (n-1) /2}t),
\end{equation}
which was picked for the special property that if $q({}\cdot{})$ is a solution then also $q_\lambda({}\cdot{})$~is. The associated constant of motion associated to~\eqref{familyWithDynamicalSymmetry} was~\eqref{constantLane-Emden} multiplied by $(n-1)/2$. It was a surprise to discover that an equivalent constant of motion (or even better for $n=1$) could be obtained with somewhat simpler calculation through the unremarkable family~\eqref{familyWithDynamicalSymmetry}. We would vaguely expect that the dynamical symmetry property of family~\eqref{familyWithDynamicalSymmetry} would give it an edge over~\eqref{thirdFamilyForLaneEmden}, but this seems not to be the case.


\section{Maxwell-Bloch system}
\label{MBsection}

The Maxwell-Bloch equations were introduced to model laser optics~\cite{Allen&Eberly} by Lamb~\cite{Lamb} (in a cavity) in 1964 and by Arecchi and Bonifacio~\cite{AB} (for free propagation) in 1965, but have interesting features from a general Dynamical Systems point of view. The conservative MB-equations with the rotating wave approximation are the following 5-dimensional system
\begin{equation}\label{MB-5}
  \dot x_1=y_1,\quad
  \dot y_1=x_1 z,\quad
  \dot x_2=y_2,\quad
  \dot y_2=x_2z,\quad
  \dot z=-(x_1y_1+x_2y_2).
\end{equation}
Authors that have recently written about this system are Huang~\cite{Huang}, Birtea and Ca\c{s}u \cite{Birtea&Casu}, 
and Ca\c{s}u~\cite{Casu}.

We are particularly interested in this last paper~\cite{Casu}, where equations~\eqref{MB-5} are embedded through
\begin{equation}\label{transf}
  q_1=x_1,\quad \dot q_1=y_1,\quad q_2=x_2,\quad
  \dot q_2=y_2,\quad \dot q_3=z
\end{equation}
into the following Lagrangian system in dimension~6
\begin{gather}\label{LagrangianMB}
  L=\frac12
  \bigl(\dot q_1^2+\dot q_2^2+\dot q_3^2+\dot 
    q_3\bigl(q_1^2+q_2^2)\bigr),\\
  \label{MB}
  \ddot q_1=q_1\dot q_3,\qquad \ddot q_2=q_2\dot q_3,\qquad
  \ddot q_3=-(q_1\dot q_1+q_2\dot q_2),
\end{gather}
for which three independent first integrals are known
\begin{equation}\label{firstIntegralsForMaxwellBloch}
  E=\frac12(\dot q_1^2+\dot q_2^2+\dot q_3^2),\quad
  B=\dot q_3+\frac12(q_1^2+q_2^2),\quad
  J=q_1\dot q_2-q_2\dot q_1.
\end{equation}
These are easily deduced from Noether's theorem because of three obvious symmetries of the Lagrangian: $E$~comes from the fact that the Lagrangian is autonomous, $B$~from the fact that~$L$ does not depend on~$q_3$, $J$~from the invariance of~$L$ under rotations in the $(q_1,q_2)$ plane. The first integral called~$H$ in formula~(15) of Huang~\cite{Huang} reduces to~$E$.

A useful nonlocal constant of motion for the Lagrangian system~\eqref{LagrangianMB} arises from the non-uniform scaling family
\begin{equation}\label{familyMB}
  q_\lambda(t):=\bigl(e^{\lambda}q_1(t), 
  e^{\lambda}q_2(t), e^{a\lambda}q_3(t)\bigr)
\end{equation}
where~$a$ is a parameter. We compute 
\begin{multline}
  \frac{\partial}{\partial\lambda}
  L\bigl(t,q_\lambda(t),\dot q_\lambda(t)\bigr)
  \Big|_{\lambda=0}
  =\\
  =\dot q_1(t)^2+\dot q_2(t)^2
  +a\dot q_3(t)^2+\Bigl(1+\frac{a}{2}\Bigr)
  \dot q_3(t)\bigl(q_1(t)^2+q_2(t)^2\bigr).
\end{multline}
The choice $a=-2$ simplifies the formula:
\begin{equation*}
  \frac{\partial}{\partial\lambda}
  L\bigl(t,q_\lambda(t),\dot q_\lambda(t)\bigr)
  \Big|_{\lambda=0}
  =\dot q_1(t)^2+\dot q_2(t)^2-2\dot q_3(t)^2
  =2E-3\dot q_3(t)^2.
  \end{equation*}
The associated constant of motion is  
\begin{equation}\label{nonlocalConstantForMaxwellBloch}
  (\dot q_1(t), \dot q_2(t), B)
  \cdot(q_1(t),q_2(t),-2q_3(t)) -2Et
  +3\int^t_{t_0}\dot q_3(s)^2 ds.
\end{equation}
A first consequence is that the following function decreases along the solutions
\begin{equation}
\frac12\frac{d}{dt}\bigl(q_1(t)^2+q_2(t)^2\bigr)-2Bq_3(t)-2Et.
\end{equation}
More interestingly, using~\eqref{MB} the constant of motion~\eqref{nonlocalConstantForMaxwellBloch} can be rewritten as
\begin{equation}
  -\ddot q_3(t)-2Bq_3(t)-2Et+3\int_{t_0}^t \dot q_3(s)^2ds
  =\text{constant}
\end{equation}
which only features~$q_3$ explicitly. By taking the time derivative, this turns into a differential equation of order~3 for~$q_3$:
\begin{equation}\label{thirdorder}
  -\dddot q_3(t)-2B\dot q_3(t)-2E+3\dot q_3(t)^2=0.
\end{equation}
Of course, the coefficients $B,E$ depend on the initial conditions. Multiplying by $\ddot q_3$ and integrating we get a constant of motion containing second time derivatives:
\begin{equation}\label{constantOfMotionWithq3dotdot}
  \frac12\ddot q_3^2+2E\dot q_3+B\dot q_3^2-\dot q_3^3.
\end{equation}
Using \eqref{MB} and~\eqref{firstIntegralsForMaxwellBloch} we can rewrite~\eqref{constantOfMotionWithq3dotdot} as a true first integral, which however is a function of $E,B,J$:\begin{equation}\label{newFirstIntegral}
  K=\frac12(q_1\dot q_1+q_2\dot q_2)^2+
  (\dot q_1^2+\dot q_2^2+\dot q_3^2)\dot q_3+
  \frac12(q_1^2+q_2^2)\dot q_3^2=
  2BE-\frac{1}{2}J^2.
\end{equation}
The constant of motion~\eqref{constantOfMotionWithq3dotdot} can be used to solve for $\dot q_3$ by quadrature:
\begin{equation}
  \pm\frac{\ddot q_3}{\sqrt{2}
  \sqrt{\dot q_3^3-B\dot q_3^2-2E\dot q_3+K}}
  =1.
\end{equation}
Define the function $\varphi$ as
\begin{equation}\label{specialFunction}
  \varphi_{E,B,K}(u)=\int
  \frac{du}{\sqrt{2}\sqrt{u^3-Bu^2-2Eu+K}}
\end{equation}
(it can be expressed in terms of special elliptic functions). Then the following is constant of motion dependent on~$t$:
\begin{equation}\label{ellipticConstantOfMotion}
  \pm\varphi_{E,B,K}(\dot q_3)-t.
\end{equation}
The $\pm$ is the sign of $\ddot q_3=-(q_1\dot q_1+q_2\dot q_2)$.

The constant of motion~\eqref{constantOfMotionWithq3dotdot} helps visualize the qualitative behaviour of $\dot q_3(t)$, because the trajectories of the couple $(\dot q_3,\ddot q_3)$ must lie in level sets of the polynomial function of two variables
\begin{equation}
  \psi_{E,B}(u,v):=\frac{1}{2}v^2+2Eu+Bu^2-u^3,
\end{equation}
which is simply formula~\eqref{constantOfMotionWithq3dotdot}, interpreted as a function of the two independent variables~$\dot q_3,\ddot q_3$, with $E,B$ treated as simple constants. More precisely, given initial data $q_1(0),\dot q_1(0),q_2(0),\dot q_2(0),q_3(0),\dot q_3(0)$ we can calculate the specific values of $E,B,K$, and these will determine the particular level set $\psi_{E,B}=K$ to which the point $(\dot q_3(t),\ddot q_3(t))$ will belong for all~$t$. Parts of the level set may easily be seen as inaccessible, because they lie either in a different connected component, or outside the stripe $\lvert\dot q_3\rvert\le(2E)^{1/2}$ imposed by the conservation of~$E$. Let us give a couple of specific examples without attempting a full analysis.

In Figure~\ref{genericOrbit} the thick closed line is the level set corresponding to the initial data $q_1(0)=\dot q_1(0)=q_2(0)=\dot q_2(0)=1$, $\dot q_3(0)=-2$, whence $\ddot q_3(0)= -q_1(0)\dot q_1(0)-q_2(0)\dot q_2(0)=-2$. The dashed line to the right belongs to the same level set $\psi_{E,B}=K$, but is not visited by the solution. The thin curved lines are other level sets of $\psi_{E,B}$.

The closed curve can degenerate into an equilibrium. For instance,  the initial data $q_1(0)=1$, $q_2(0)=\dot q_1(0)=3$, $\dot q_2(0)=-1$, $\dot q_3(0)=-1$ give the equilibrium $(\dot q_3,\ddot q_3)=(-1,0)$ of the third order equation \eqref{thirdorder}. The full solution is
\begin{equation}\label{stationarySolution}
  q_1(t)=\cos t+3\sin t,\qquad
  q_2(t)=3\cos t-\sin t,\qquad
  \dot q_3(t)=-1.
\end{equation}

The level set $\psi_{E,B}(u,v)=K$ is a nodal cubic for $B>0,E=B^2/2$ and $K=B^3$ (so $J=0$) as for example for $q_1(0)=2$, $\dot q_1(0)=q_2(0)=\dot q_2(0)=0$, $\dot q_3(0)=-1$. Figure~\ref{homoclinicOrbit} shows the corresponding homoclinic orbit. In this case the integral~\eqref{specialFunction} simplifies and we can find $\dot q_3$ explicitly
\begin{equation}\label{homoclinicSolution}
  \dot q_3(t)=1-2 \sech^2t.
\end{equation}
 Plugging this function into the first two equations \eqref{MB} with the initial data above and by the transformation \eqref{transf}, we get the homoclinic solution to the 5-dimensional system \eqref{MB-5}
\begin{equation}
x_1(t)=2\sech t,\ y_1(t)=-2\sech t\tanh t,\  x_2(t)=y_2(t)=0,\  z(t)=1-2 \sech^2t.
\end{equation}
 Such homoclinic  solutions were found by Holm,  Kovacic, Sundaram~\cite{Holm}, Huang \cite{Huang}, and Birtea, Ca\c su~\cite{Birtea&Casu} with  different methods.

\begin{figure}
\begin{center}
\includegraphics[height=.4\textheight]{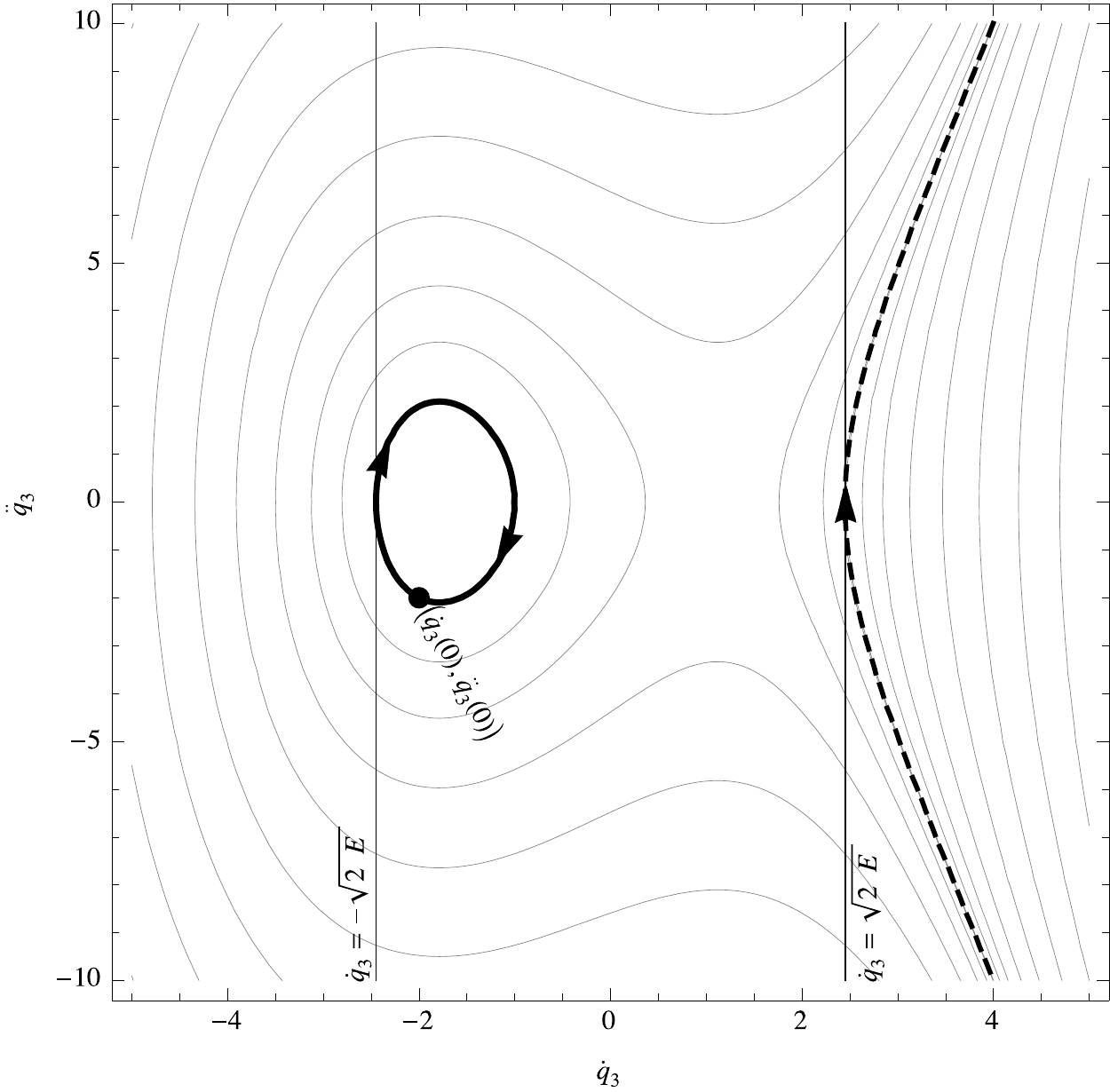}
\end{center}
\caption{A generic orbit for $(\dot q_2,\ddot q_3)$ embedded in the level sets of $\psi_{E,B}$}
\label{genericOrbit}
\bigskip
  \begin{center}
\includegraphics[height=.4\textheight]{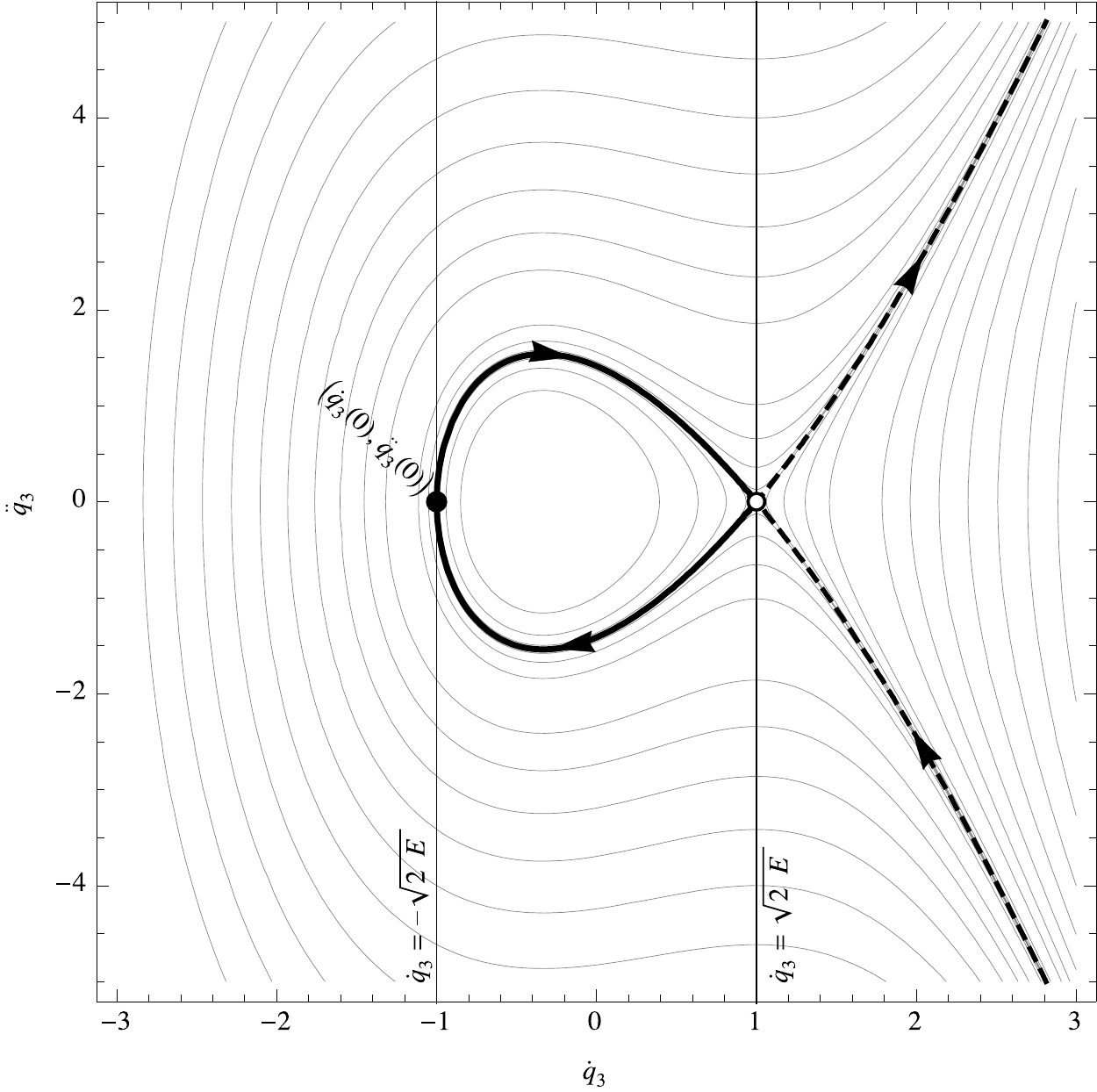}
\end{center}
\caption{Homoclinic orbit}
\label{homoclinicOrbit}
\end{figure}

\section*{Acknowledgment}

The second author did part of the work for this paper while visiting Universitat Aut\`{o}noma de Barcelona. He is thankful for the warm hospitality of Armengol Gasull. We thank Nicola Sansonetto for some useful comments.

\section*{Keywords}

Lagrangian ODE; nonlocal constants of motion; Maxwell-Bloch with RWA; Lane-Emden; dissipative mechanical systems.


\end{document}